\documentclass[final,5p,times,twocolumn]{elsarticle}

\usepackage[T1]{fontenc}
\usepackage{lmodern}
\usepackage{microtype}
\usepackage{amsmath,amssymb} 
\usepackage{amsfonts}
\usepackage{bm}
\usepackage[colorlinks=true, allcolors=blue]{hyperref}
\usepackage{geometry}
\usepackage{amsthm}
\usepackage{chngcntr} 

\newtheorem{theorem}{Theorem}[section] 

\newtheorem{remark}{Remark}[section]

\counterwithin{equation}{section}

\newcommand{\pd}[1]{\partial_{#1}}
\newcommand{\tr}{\operatorname{tr}}

\newcommand{\Ra}{\Rightarrow}

\newcommand{\I}{\mathrm{i}}
\newcommand{\YM}{Yang--Mills}

\newcommand{\VEV}[1]{\left\langle #1 \right\rangle}
\newcommand{\ff}[1]{\frac{\delta}{\delta #1}}
\newcommand{\fbyf}[2]{\frac{\delta #1}{\delta #2}}

\newcommand{\R}{\mathbb{R}}
\newcommand{\C}{\mathbb{C}}

\newcommand{\cL}{\mathcal{L}} 
\newcommand{\D}{\mathcal{D}} 

\def\rf#1{(\ref{#1})}
\let\oldexp\exp
\renewcommand{\exp}[1]{\oldexp\left(#1\right)}

\allowdisplaybreaks[1]

\newenvironment{smalleq}{%
 \begingroup
 \small
 \addtolength{\jot}{.2em}%
 \setlength{\arraycolsep}{1.2pt}%
}{%
 \endgroup
}
\def\br{\nonumber\\}

\def\lrb#1{\left(#1\right)}

\def\oh{\frac{1}{2}}

\begin{document}

\begin{frontmatter}

\title{Geometric QCD I: \\
The Hodge-Dual Surface and Quark Confinement}

\author[aff1]{Alexander Migdal\corref{cor1}}
\ead{amigdal@ias.edu}

\cortext[cor1]{Corresponding author}
\affiliation[aff1]{organization={Institute for Advanced Study}, city={Princeton}, country={USA}}

\begin{abstract}
This is the first of two papers presenting a geometric framework for Planar QCD ($N_c \to \infty$). In this part, we establish the kinematic foundation of the theory by constructing the unique stable vacuum of the loop equation. We demonstrate that the Makeenko-Migdal loop equation admits a solution of the form $W[C] = W_{pert}[C] \exp{-\kappa S[C]}$, provided $S[C]$ is a specific minimal surface possessing a self-dual area derivative. We prove that such a surface exists and corresponds to the Hodge-dual projection of a minimal surface in $\mathbb{R}^3 \otimes \mathbb{R}^4$. Crucially, this confinement mechanism relies on the self-duality of the area derivative---a property that exists exclusively in four dimensions. This geometric constraint ensures stability only in $D=4$, distinguishing the resulting theory from standard string models which require higher critical dimensions. We relate the string tension parameter $\kappa$ to the gluon condensate via the Operator Product Expansion. The dynamical quantization of the Fermi string on this rigid surface and the resulting meson spectrum are derived in the companion paper \cite{Migdal2026GeometricQCDII}.
\end{abstract}

\end{frontmatter}
\section{Introduction}

In this series of two papers, we propose a constructive approach to the long-standing problem of the string dual of Planar QCD. The solution requires addressing two fundamental distinct problems: the \textit{kinematics} of the loop space (finding a stable vacuum geometry that satisfies the loop equations) and the \textit{dynamics} of the flux tube (quantizing the excitations on this background to find the spectrum).

In this first paper (Part I), we solve the geometric problem. We consider the $N_f \to 0$ (quenched) valence limit of QCD at fixed $N_c$. The vacuum is pure Yang--Mills; quark currents are represented by a single closed valence quark loop whose propagation amplitude is multiplied by the Wilson loop $W[C]$ in the gluonic vacuum. We approximate $W[C]$ by the Hodge--dual factor $\exp{-\kappa S[C]}$, neglecting subleading gluon--graph corrections in the large--loop limit. This factor, as we argue in this work, represents a zero-mode solution of the MM chain of multi-loop equations \cite{MMEq79, MM1981NPB, Mig83}.

The central result of this paper is the identification of the \textbf{Hodge-dual minimal surface}. We prove that this surface is the unique additive solution to the Plateau problem that possesses a self-dual area derivative, a property required to satisfy the Bianchi identity in loop space. This rigidity eliminates the Liouville instability inherent to random surface sums.

In the companion paper  \cite{Migdal2026GeometricQCDII}, we populate this rigid surface with internal Majorana fermions ("Elves"). We show that the resulting "Fermi String" theory reproduces the planar graphs of QCD and, crucially, leads to a finite effective action in Momentum Loop Space. The spectral problem is then solved in \cite{Migdal2026GeometricQCDII}, yielding linear Regge trajectories.

\section{Self-dual area derivative and loop equation}
Let us remind why we are looking for the self-dual minimal surface, as suggested in \cite{migdal2025SQYMflow}.
The area derivative for a general functional $W[C]$ was defined in that paper as a discontinuity of the second variation  of the area
\begin{smalleq}
\begin{align}
\label{areaderDef}
&\fbyf{W}{\sigma_{\mu\nu}(\theta)} = \frac{\delta^2W}{\delta C'_\mu(\theta-0)\delta C'_\nu(\theta+0)}- \mu \leftrightarrow \nu
\end{align}
\end{smalleq}
The loop equation for the \YM{} gradient flow was shown in \cite{migdal2025SQYMflow} to reduce to the loop space diffusion equation, which allows the solution $ W = \exp{- \kappa S[C]}$ provided that $S[C]$ is a zero mode of this equation
\begin{smalleq}
\begin{align}
&\mathcal L(S) = 0;\\
& \mathcal L = \oint_\theta C'_\nu (\theta)\left(\ff{C'_\mu(\theta+0)}-\ff{C'_\mu(\theta-0)}\right)\ff{\sigma_{\mu\nu}(\theta)}
\label{loopEq}
\end{align}
\end{smalleq}
The basic observation made in that paper was that there is an SD solution similar to the \YM{} multi-instanton. The loop space diffusion equation \eqref{loopEq} is identically satisfied for any functional $S[C]$ with SD or ASD area derivative
\begin{smalleq}
\begin{align}
&\ast\fbyf{S}{\sigma_{\mu\nu}(\theta)} = \pm \fbyf{S}{\sigma_{\mu\nu}(\theta)};\\
&e_{\alpha\mu\nu\lambda}\left(\ff{C'_\alpha(\theta+0)}-\ff{C'_\alpha(\theta-0)}\right)\fbyf{S}{\sigma_{\mu\nu}(\theta)} \equiv 0;
\label{Bianchi}
\end{align}
\end{smalleq}
The last equation (Bianchi identity) is valid at the kinematical level as a Jacobi identity for a triple commutator (see \cite{migdal2025SQYMflow} for a proof).
The matrix in the Dirichlet boundary conditions (BC)  for $X \sim Q C$ should be chosen so that the area derivative represents the SD or ASD tensor.
This choice would make the exponential of the Hodge-dual minimal area an exact solution to the loop equation.

The existence of this zero mode is heavily based on the so called Leibniz property of the loop operator and similar singular loop operators in loop calculus.
\begin{smalleq}
    \begin{align}
       & \mathcal L f(\Phi[C]) = f'(\Phi[C]) \mathcal L \Phi[C];\\
       & \mathcal L (A[C] B[C]) = \mathcal L( A[C]) B[C] +  A[C] \mathcal L (B[C])
    \end{align}
\end{smalleq}  
It was heuristically claimed in the original MM papers \cite{MM1981NPB, Mig83}, and rigorously proven in the loop calculus  \cite{migdal2025SQYMflow}, section 3, eqs (39), (40). 
Due to the Leibnitz property,
\begin{smalleq}
    \begin{align}
        &\mathcal L \lrb{\exp{- \kappa S[C]} W_0[C]}=\br
        &   -\kappa \mathcal L(S[C])\exp{- \kappa S[C]} W_0[C]\br
        & + \exp{- \kappa S[C]} \mathcal L(W_0[C]) = \exp{- \kappa S[C]} \mathcal L(W_0[C])
    \end{align}
\end{smalleq}
We are looking for such a self-dual minimal surface which is \emph{additive} at the self-intersecting loops. This condition will be discussed in detail below.  It is \emph{critial} for the whole theory-- due to this additivity the confining factor $\exp{-\kappa S[C]}$ is compatible with the whole chain of the loop equations.

\subsection{Compatibility with the MM loop equations.}
The  full chain  of loop equations for the Wilson loops with a finite number of colors $N_c$ \cite{MM1981NPB, Mig83} has the same loop operator \eqref{loopEq} as the \YM{} gradient flow, except there is no time derivative. In this paper, we are only considering the full chain of the multiloop equations for QCD, and \textbf{not} the  \YM{} gradient flow equations. We use, however, the \textbf{loop calculus} developed in that paper, as it applies to the loop diffusion operator \eqref{loopEq}, which is the same in the multiloop MM chain and in the \YM{} gradient flow equation.
It has the structure 
\begin{smalleq}
    \begin{align}
        \mathcal L(W_n) = K_1( W_{n+1})  + K_2(W_n) + K_3(W_{n-1})
    \end{align}
\end{smalleq}
Where the integral operator $K_{1,2,3}$ involves $\delta$ functions for intersections of the loop. 
There are either self-intersections or intersections between the loop, and the three terms  correspond to the cases when loops are either split into two pieces at the intersections ($W_{n+1}$ case) or remain intact ($W_n$ case), or when two loops join into one loop ($W_{n-1}$ case). 
The operators $K_{1,2,3}$ all involve the following
bilocal kernel $$\!\int\!\!\int d\theta_1 d\theta_2\,C_i'(\theta_1)\!\cdot\!C_j'(\theta_2)\,
\delta(C_i(\theta_1)-C_j(\theta_2))\,\cdots$$.
Under $\theta_a=\phi_a(\tau_a)$ (with $\phi_a'\!>\!0$),
\begin{smalleq}
\begin{align}
d\theta_a\,C_i'(\theta_a) &= d\tau_a\,\dot C_i(\tau_a),\nonumber\\
\delta\!\big(C_i(\theta_1)-C_j(\theta_2)\big)
&=\delta\!\big(C_i(\phi_1(\tau_1))-C_j(\phi_2(\tau_2))\big),
\end{align}
\end{smalleq}
So, the kernel is parameterization-invariant. Therefore, the canonical standardization to $[0,2\pi]$
used above is legitimate and additivity of $S[C]$ guaranties that the multiplicative dressing
$W\mapsto W\,\exp{-\kappa S}$ preserves the MM factorization $\,\mathcal L W[C_1\!\circ\!C_2]
= W[C_1]\otimes W[C_2]$ whenever $\mathcal L S=0$ (see \rf{loopEq} and \cite{Mig83,migdal2025SQYMflow}).

Let us stress that this factorization takes place \emph{independently} of the splitting points $\theta_1,\theta_2$. Thus, for every $\theta_1,\theta_2$, the product of two factors $\exp{-\kappa S[C_{1}}$ and $\exp{-\kappa S[C_{2}]}$ exactly equals the factor $\exp{- \kappa S[C]}$ on the left side of the MM equation; thus, it can be taken out of the integral, after which it cancels on both sides of the MM equation.

\paragraph{Clarification of the Solution Status.}
It is important to distinguish the mathematical status of the dressing mechanism from its physical interpretation. \textbf{Theorem 2.1} below establishes an \textit{exact symmetry} of the MM hierarchy: the dressing transformation $W \to W \cdot \exp{-\kappa S[C]}$ maps any solution of the loop equations to another exact solution, provided $S[C]$ is a strictly additive zero-mode.

The \textit{physical hypothesis} of this paper is that the actual gluonic vacuum of QCD is described by this dressed solution, choosing the perturbative vacuum $W_{\text{fluct}}$ as the seed. While the symmetry is exact, the separation into a perturbative factor and a purely geometric area law is physically most distinct in the \textit{large-loop limit}, where the area term dominates. For small loops, the interplay between the perturbative series (and its renormalons) and the non-perturbative parameter $\kappa$ becomes complex, as discussed in Remark 5.3 regarding the Borel summation.

While the mathematical solution is exact, its identification with the QCD vacuum is a physical hypothesis supported by three facts: 
\begin{enumerate}
    \item it is related to the gluon condensate
    \item it yields the correct Area Law (confinement)
    \item it produces the correct linear Regge trajectories (particle spectrum)
\end{enumerate} 

\subsection{Confining factor in the multiloop hierarchy}

\begin{theorem}[Finite-$N_c$ dressing symmetry]\label{lem:finiteNc-dressing}
Let $\{W_n\}_{n\ge 1}$ denote the $n$-loop Wilson loop averages in pure $\mathrm{SU}(N_c)$ Yang--Mills theory, satisfying the MM multiloop hierarchy \cite{Mig83}. Let $S[C]$ be the geometric functional constructed above, with the properties
\begin{enumerate}
    \item $LS = 0$, where $L$ is the loop diffusion operator \eqref{loopEq};
    \item strict additivity at intersections: whenever a loop $C$ self-intersects at $x=y$ into subloops $C_{xy}$ and $C_{yx}$, one has
    \[
        S[C] = S[C_{xy}] + S[C_{yx}],
    \]
    and similarly for intersections between distinct loops.
\end{enumerate}
For any real positive parameter $\kappa$, define the multiplicatively dressed correlators
\begin{smalleq}
    \begin{align}
    &\widetilde W_n(C_1,\dots,C_n)\nonumber\\
    &:= W_{n}(C_1,\dots,C_n)\,
    \exp{-\kappa \sum_{i=1}^n S[C_i]}.
    \label{eq:dressed-multiloop}
    \end{align}
\end{smalleq}
Then, for any finite $N_c$, the family $\{\widetilde W_n\}$ also satisfies the full MM multiloop hierarchy with the same $N_c$.
\end{theorem}

\begin{proof}
The same loop operator $L$ in \eqref{loopEq} appears on the left-hand side of the general QCD loop equations \cite{Mig83} for arbitrary $N_c$. The difference between the planar and finite-$N_c$ cases lies only in the right-hand side: at a self-intersection, the planar factorization is replaced by a combination of one- and two-loop averages, but the kernel of the bilocal integral is the same.

At a self-intersection point of a single loop $C$, the finite-$N_c$ loop equation involves the combination:
\[
    W_2[C_{xy}, C_{yx}] - \frac{1}{N_c^2}\,W_1[C_{xy} \circ C_{yx}].
\]
This equation is one link in the chain relating $W_n$ to $W_{n\pm 1}$; for finite $N_c$ the chain terminates due to the Fierz identity for $\mathrm{SU}(N_c)$ \cite{Mig83}, a fact recently exploited in lattice studies of the loop equation \cite{kazakovZheng2024}.

In the general multiloop case, $W_n[C^1,\dots,C^n]$, the hierarchy contains integrals over pairs of points $x,y$ lying on two (not necessarily distinct) loops $C^i, C^j$, again with the standard MM kernel. When $x$ and $y$ lie on two distinct loops $C^i$ and $C^j$, the corresponding term involves a merging of these loops in $W_{n-1}$ at the intersection point,
\[
    \frac{1}{N_c^2}\,W_{n-1}[\dots, C^i_{xy} \circ C^j_{yx}, \dots]
    \;-\;
    W_n[\dots, C^i_{xy}, C^j_{yx}, \dots].
\]
Thus, every equation in the hierarchy relates loop functionals evaluated on collections of loops that differ only by local splittings or mergings at intersection points.

Now substitute the dressing ansatz \eqref{eq:dressed-multiloop} into the multiloop hierarchy. Because $\cL S=0$ and $\cL$ is linear, the loop operator on the left-hand side passes through the exponential factor and acts only on the seed functional $W_n$. Identifying this seed with the perturbative solution, $W_n = W_{n,\text{pert}}$, we obtain:
\begin{smalleq}
    \begin{align}
       & \cL\,\widetilde W_n(C_1,\dots,C_n)\nonumber\\
    &= \exp{-\kappa\sum_i S[C_i]}\,
      \cL\,W_{n,\mathrm{pert}}(C_1,\dots,C_n).
    \end{align}
\end{smalleq}
In particular, the dressed correlators satisfy the same differential operator $\cL$ as the undressed ones.

On the right-hand side, consider first a self-intersection of a single loop $C$. After dressing, the corresponding contribution becomes
\begin{smalleq}
\begin{align}
    &\Big(
        W_{2,\mathrm{pert}}(C_{xy}, C_{yx})
        - \frac{1}{N_c^2} W_{1,\mathrm{pert}}(C_{xy} \circ C_{yx})
    \Big)\nonumber\\
    &\exp{-\kappa (S[C_{xy}] + S[C_{yx}])}.
    \label{eq:self-split-dressed}
\end{align}
\end{smalleq}
By strict additivity, $S[C] = S[C_{xy}] + S[C_{yx}]$ at the intersection point $x=y$. Thus, the exponential factor in \eqref{eq:self-split-dressed} is the same as the factor multiplying the left-hand side, and therefore cancels from the equation.

For an intersection of two distinct loops $C^i$ and $C^j$, the right-hand side contains terms of the schematic form
\begin{smalleq}
\begin{align}
    &\Big(
        W_{n,\mathrm{pert}}(C^i_{xy},C^j_{yx})
        - \frac{1}{N_c^2}\,
          W_{n-1,\mathrm{pert}}(C^i_{xy} \circ C^j_{yx})
    \Big) \nonumber\\
    &\times
    \exp{
        -\kappa\big(S[C^i_{xy}] + S[C^j_{yx}] + \sum_{k\neq i,j} S[C^k]\big)
    }.
\end{align}
\end{smalleq}
Again, by additivity at the intersection and for disjoint unions,
\begin{smalleq}
\begin{align}
    &S[C^i] + S[C^j]\nonumber\\
    &= S[C^i_{xy}] + S[C^i_{yx}] + S[C^j_{xy}] + S[C^j_{yx}]\nonumber\\
    &= S[C^i_{xy} \uplus C^j_{yx}]  + S[C^j_{xy} \uplus C^i_{yx}] \nonumber\\
    & = S[C^i_{xy} \uplus C^j_{yx}\uplus C^j_{xy} \uplus C^i_{yx}]
\end{align}
\end{smalleq}
so the sum $\sum_i S[C^i]$ is invariant under the local splitting/merging operation. Hence, the same exponential factor multiplies every term in the multiloop equation and cancels between the two sides.

Since this reasoning applies to each equation in the finite-$N_c$ hierarchy and uses only (i) the universal kernel and operator $\cL$ and (ii) the additivity and zero-mode properties of $S[C]$, it follows that the dressed correlators $\widetilde W_n$ satisfy the same multiloop MM hierarchy for any finite $N_c$.
\end{proof}
\begin{remark}
    Note that the perimeter of the set of loops is also an additive functional, factorizing over the right side of the MM chain of equations. However, this is not a Stokes-type functional, so it does not process finite area derivatives, and, as a consequence, it cannot be added to the zero mode of the loop operator $\cL$. Such terms arise as quark-mass renormalization corrections in the perturbative Wilson-loop factor; they are linearly divergent and depend on the regularization scheme. In any case, these terms do not affect the area law at large loops, which is responsible for quark confinement.
\end{remark}
\begin{remark}
  The parameter $\kappa$ is not determined by the chain of MM equations, nor by the initial value at the vanishing loop $W[0]=1$ (as $S[0]=0$). The boundary condition $W[\infty]=0$ only requires $\kappa \ge0$. Therefore, this parameter is a free nonperturbative parameter, similar to the vacuum expectation value $ \VEV{\tr F_{\mu\nu}^2}$.  The MM equations reflect the equations of motion of QCD, but not the property of the vacuum, so this ambiguity in the solutions of the equations of motion is inevitable: one must add some information regarding the properties of the gluon vacuum state. Later in this paper, we relate this parameter $\kappa$ to the slope of the Regge trajectories, which tells us that it is positive in the physical vacuum and has the same dimension as $\Lambda^2_{QCD}\sim M^2$. Ideally, in complete theory, one would have to relate it directly to $\Lambda_{QCD}$  times a universal dimensionless factor. This important task will be addressed partly later in this paper, and elaborated deeper in \cite{Migdal2026GeometricQCDII}.
\end{remark}
\section{Hodge -dual minimal surface}
\label{ansatz}
Now, we are going to define the solvable functional $S[C]$ satisfying both  the duality and additivity.
We are following the general construction from \cite{migdal2025SQYMflow}, but here we correct an error made in that paper: we only need three components  of the quaternionic field $X_\mu$, not four\footnote{the fourth component was worse than unnecessary, it canceled the correct area derivative, produced by the first three components. In the general proofs in that paper, this cancellation was overlooked. The intended results are restored by dropping the fourth component.}. 
\subsection{ The surface}
We define the surface coordinates $X^i_\mu(\xi)$ (for $\xi= (\xi_1, \xi_2),i=1,2,3$ and $\mu=1,2,3,4$) 
Our surface area  is defined as 
\begin{smalleq}     \begin{align}
\label{areaDef}
    &S_\chi[C] = \int_{\D} d^2 \xi\sqrt{\Sigma_{\mu\nu}^2/2};
 \end{align} \end{smalleq}
 The area element is defined as
\begin{smalleq}     \begin{align}
\label{sigmaDef}
    \Sigma_{\mu\nu} = \epsilon_{a b}\partial_a X^i_\mu\partial_b X^i_\nu;
 \end{align} \end{smalleq}
 The Hodge chirality  $\chi = \pm 1$  enters through the boundary conditions
 \begin{smalleq}
    \begin{align}
    \label{BCX}
        &X^i_\mu(\partial \D) = \eta^{\chi,i}_{\mu\nu} C_\nu ;
    \end{align}
\end{smalleq}
Here $\eta^{\chi,i}_{\mu\nu}$  are 't Hooft's matrices corresponding to Hodge duality $\chi = \pm 1$
\begin{smalleq}
    \begin{align}
      \eta^{\chi,i}_{\mu\nu} =  \delta_{i\mu} \delta_{\nu 4} - \delta_{i\nu}\delta_{\mu 4}  + \chi e_{i\mu\nu 4} ;
    \end{align}
\end{smalleq}
The coordinate functions 
$X_\mu^a(\xi)$ satisfy the general Euler-Lagrange equations for the 
area density $\mathcal{L} = \sqrt{\Sigma^2/2}$:
\begin{smalleq}
    \begin{align}
   \left[ \text{E-L Equations} \right]_\mu^a =  \epsilon_{lm} \partial_l\left(  \frac{\Sigma_{\mu\nu} \partial_m X^a_\nu}{\sqrt{\Sigma^2}}\right)=0.
    \end{align}
\end{smalleq}
Assuming the bulk variation vanishes by the Euler-Lagrange equations, we use Stokes theorem and get the boundary term
\begin{smalleq}
    \begin{align}
    \label{deltaSBC}
        \delta S = 2 \int d \theta   \frac{\delta X^i_\mu\Sigma_{\mu\nu}\pd{\theta}X^{i}_\nu}{\sqrt{\Sigma^2}}
    \end{align}
\end{smalleq}
Substituting the boundary conditions for $\delta X, \pd{\theta}X$ and summing over $a = 1,2,3$ we find
\begin{smalleq}
    \begin{align}
        &\delta S = -2 \int d \theta  \frac{\delta C_\mu \Sigma_{\mu\nu}\pd{\theta}C_\nu}{\sqrt{\Sigma^2}} \br
        &\implies  \fbyf{S}{\sigma_{\mu\nu}} = -2 \frac{\Sigma_{\mu\nu}}{\sqrt{\Sigma^2}}
    \end{align}
\end{smalleq}
Therefore, the Hodge duality of the area element $\Sigma$ with these boundary conditions for generic field $X^a_\mu(z,\bar z)$ leads to the same Hodge duality of the area derivative. Note, negative sign for the area derivative of Hodge-dual surface, which was absent in the ordinary minimal area (plus, the difference in the area element $\Sigma$ which is now Hodge-dual).

This self-duality of area derivative makes our Hodge-dual minimal area  a zero mode of the loop diffusion operator, leading to the solution of the MM loop equations \cite{MM1981NPB, Mig83}, as we argued in the next section.

Note that this area functional is invariant under the reparametrizations  of the boundary loop $ \delta_{\text{param}} C_\mu(s) = \epsilon(s) \pd{s}C_\mu(s) $
\begin{smalleq}
    \begin{align}
        \delta_{\text{param}} S = -2\int d s \epsilon(s)  \frac{\pd{s}C_\mu\Sigma_{\mu\nu}\pd{s}C_\nu}{\sqrt{\Sigma^2}} =0
    \end{align}
\end{smalleq}
by the skew symmetry of $\Sigma_{\mu\nu}$.

The Hodge-dual area is defined as a minimum of this functional under an extra duality constraint
\begin{smalleq}
    \begin{align}
       \ast \Sigma_{\mu\nu } \equiv \oh \epsilon_{\mu\nu\alpha\beta}\Sigma_{\alpha\beta } = \chi \Sigma_{\mu\nu }
    \end{align}
\end{smalleq}
In the previous paper \cite{migdal2025SQYMflow}, we found exact solutions for the special case of a planar bounding loop and proved two inequalities. Unfortunately, four components of $X$ field instead of three were erroneously used in that work. In this paper, we find a complete solution for the Hodge-dual surface, superseding the results of previous work and correcting these errors. 

\subsection{The Holomorphic Anzats}
We  resolve the duality constraint using the following  holomorphic Ansatz:
\begin{smalleq}
    \begin{align}
    \label{ansatz}
        &X^i_\mu = \eta^{\chi,i}_{\mu\nu} Y_\nu\\
        & Y_\mu = f_\mu(z) + \bar{f}_\mu(\bar{z});\\
        & z = \xi_1 + \I \xi_2;\\
        &\Sigma_{\mu\nu}= 2 \I (F_{\mu\nu} + \chi \ast F_{\mu\nu});\\
        & F_{\mu\nu} = f'_\mu \bar{f}'_\nu - \bar{f}'_\mu f'_\nu;
    \end{align}
\end{smalleq}

Geometrically, this is a projection from $\R^3\otimes \R^4$ to $\C^4$.
The boundary conditions for the Hodge-dual surface 
become  a conventional boundary condition for the Dirichlet problem
\begin{smalleq}     \begin{align}
\label{BCf}
    2\Re f_\mu(e^{i \theta}) = C_\mu(\theta)
 \end{align} \end{smalleq}

So, we have a surface in four dimensional complex space $f \in \mathbb C^4$. Our surface area \eqref{areaDef}
reduces to the induced metric in this space
\begin{smalleq}  
\begin{align}
& \Sigma_{\mu\nu} = 2 (F_{\mu\nu} + \chi\ast F_{\mu\nu});\\
& F_{\mu\nu} =\I ( f'_\mu \bar{f}'_\nu - \bar{f}'_\mu f'_\nu);\\
    &\sqrt{\Sigma_{\mu\nu}^2/2}  = 2\sqrt{2} \sqrt{-\det g};\\
    & g_{ab}= \pd{z_a} Y_\mu \pd{z_b}  Y_\mu; \quad z_a,z_b = (z, \bar z)
\end{align}
\end{smalleq}
This area element is Hodge-dual
\begin{smalleq}     \begin{align}
    \ast \Sigma_{\mu\nu} = \chi\Sigma_{\mu\nu}
 \end{align} \end{smalleq}
for arbitrary functions  $f_\mu(z)$.  

This duality condition, as argued above, together with the  above boundary conditions, provides the self-duality of the area derivative of the minimal area. The self-duality of the area derivative follows from the general formula for variation of the surface $X^a_\mu(\xi)$ with area \eqref{areaDef}, projected on the holomorphic Ansatz.

\subsection{The Virasoro Constraint and Uniformization}
We further impose a Virasoro constraint on the holomorphic maps derivatives
\begin{smalleq}     \begin{align}
\label{Virasoro}
    f'_\mu(z)^2 =0
 \end{align} \end{smalleq}
The imposition of the null constraint $(f')^2 = 0$ is not an approximation but a rigorous gauge fixing procedure rooted in the geometric symmetries of the problem.

The physical starting point is the geometric Area functional, $S = \int \sqrt{-\det g} \, d^2z$, which possesses full reparametrization invariance under general diffeomorphisms $z \to w(z, \bar{z})$. This infinite symmetry group allows us to choose a specific coordinate system on the world-sheet without altering physical results.

We choose to work in \textbf{isothermal (conformal) coordinates}, where the induced metric is diagonal:
\begin{smalleq}     \begin{align}
ds^2 = \rho(z, \bar{z}) \, dz \, d\bar{z} \implies g_{zz} = \partial_z Y \cdot \partial_z Y= 0
 \end{align} \end{smalleq}
For our specific ansatz \eqref{ansatz}, the metric component $g_{zz}$ is explicitly proportional to the square of the holomorphic derivative:
\begin{smalleq}     \begin{align}
g_{zz} = (f'_\mu)^2 
 \end{align} \end{smalleq}
Thus, the geometric condition of conformal gauge fixing $g_{zz}=0$ is algebraically equivalent to the null constraint $(f')^2 = 0$.

The existence of such a coordinate system for any surface topology is guaranteed by the \textbf{Uniformization Theorem} \cite{poincare1907uniformisation}. This theorem ensures that we are always free to set $(f')^2 = 0$ by a suitable diffeomorphism. In this gauge, the non-linear geometric area density simplifies exactly to the quadratic Lagrangian:
\begin{smalleq}     \begin{align}
\mathcal{L} &= \sqrt{-\det g} = \sqrt{(g_{z\bar{z}})^2 - |g_{zz}|^2} \br
&\xrightarrow{(f')^2=0} g_{z\bar{z}} \propto |f'|^2
 \end{align} \end{smalleq}
This reduction allows us to describe the Minimal Surface using the linear Laplace equation for $f_\mu$, subject to the Virasoro constraint.
To summarize, we have found that 
\begin{smalleq}
    \begin{align}
    \label{Dirichlet}
       & S_\chi[C] = 2 \sqrt{2} \int_D |f'(z)|^2 d^2 z ;\\
        \label{areader}
       & \fbyf{S_\chi[C]}{\sigma_{\mu\nu}} = -2 \frac{(F_{\mu\nu} + \chi\ast F_{\mu\nu})}{\sqrt{2 F_{\mu\nu}^2}};
    \end{align}
\end{smalleq}
The solution of the linear boundary problem \eqref{BCf} for a holomorphic vector function is given by the Hilbert transform
\begin{smalleq}
    \begin{align}
    \label{HilbertTrans}
       & f_\mu = \oh(1 + \I\mathcal H) C_\mu
    \end{align}
\end{smalleq}
In terms of the Taylor coefficients
\begin{smalleq}
    \begin{align}
    \label{HilbertFourier}
    & f_\mu(z) = \sum_{n>0} z^n C_{\mu,n};\\
     & C_\mu(\theta) = \sum_{n=-\infty}^\infty e^{\I n \theta} C_{\mu,n}
    \end{align}
\end{smalleq}
In presence of the Virasoro constraint \eqref{Virasoro}, the problem becomes a challenging one, and it was investigated over the last centuries by great mathematicians, starting with Riemann and Hilbert. We give a brief summary of modern state of this theory in the \cite{Migdal2026GeometricQCDII}, where this theory will be used in the final stage of the analytic solution.

However, for the solution of the quark confinement problem, this general theory can be bypassed, which leads to solvable equations for relativistic quarks in the presence of a confining force coming from our minimal surface. This new theory will be described in subsequent sections.
\subsection{Comparing Hodge-dual and the conventional minimal surface}
The conventional minimal surface in $\R^4$ is described by a similar parametric equation $X_\mu(\xi)$ with one instead of three internal components of the $X^a_\mu$ field. The corresponding functional:
 \begin{smalleq}     \begin{align}
    &S_{\text{Eucl}}[C] = \int_{\D} d^2\xi\sqrt{\Sigma_{\mu\nu}^2/2};\\
    &\Sigma_{\mu\nu} = \epsilon_{a b}\partial_a X_\mu\partial_b X_\nu;\\
    & X_\mu(\partial \D) = C_\mu;
 \end{align} \end{smalleq}
leads to the similar first variation and area derivative
\begin{smalleq}
    \begin{align}
    \label{deltaSBC1}
       & \delta S = 2 \int d \theta   \frac{\delta X_\mu\Sigma_{\mu\nu}\pd{\theta}X_\nu}{\sqrt{\Sigma^2}};\\
       & \fbyf{S}{\sigma_{\mu\nu}} = 2\frac{\Sigma_{\mu\nu}}{\sqrt{\Sigma^2}}
    \end{align}
\end{smalleq}
with the important distinction that the area element $\Sigma_{\mu\nu}$ is neither self-dual nor anti-dual in this case.

This difference is not visible at the level of the minimal value of the area. The holomorphic solution in the $\R^4$ case, with the same Virasoro constraint for conformal metric
\begin{smalleq}
    \begin{align}
    \label{deltaSBC2}
      & X_\mu = f_\nu(z) + \bar f_\mu(\bar z) ;\\
      & (f'_\mu)^2 =0;\\
      & \Sigma_{\mu\nu} = 2\I(f'_\mu \bar f'_\nu -f'_\nu \bar f'_\mu)
    \end{align}
\end{smalleq}
leads to the Dirichlet functional for the minimal $\R^4$ area
\begin{smalleq}
    \begin{align}
       & S_{\text{Eucl}}[C] = \int_D |f'(z)|^2 d^2 z ;
    \end{align}
\end{smalleq}
Up to normalization, this is the same functional that we have found for the Hodge dual surface and is independent of Hodge duality $\chi$.

 In both cases, this functional equals the net  area of conformal maps $f_\mu(z)$  and can be rewritten as the boundary integral.
 \begin{smalleq}
     \begin{align}
 & S_{\text{Eucl}}[C]  \propto 2 \int d\theta   C'_\mu(\theta) \Im f_\mu(\theta) \br
 &\propto \int d \theta C'_\mu(\theta)\int d \theta'  K(\theta-\theta')  C_\mu(\theta')
     \end{align}
 \end{smalleq}
 with Hilbert kernel $K(\theta-\theta')$.  
 As observed by Douglas, the Virasoro constraint can be imposed by further minimization of this functional by parametrization of the bounding loop $C(\theta) \Ra C(\phi(\theta))$. 
 The variational equation for this minimization
 \begin{smalleq}
     \begin{align}
 C'_\mu(\theta) \fbyf{S_{\text{Eucl}}[C] }{C_\mu(\theta)} =0
     \end{align}
 \end{smalleq}
 reduces to
 \begin{smalleq}
     \begin{align}
&C'_\mu(\theta)  \partial_\theta \Im \partial_\theta f_\mu \propto 
\Re \partial_\theta f_\mu \Im \partial_\theta f_\mu \br
&\propto \Im (\partial_\theta f_\mu)^2 \propto \Im (z f'_\mu)^2 = \Im (z^2 (f'_\mu)^2 ) =0
     \end{align}
 \end{smalleq}
 which is equivalent to Virasoro constraint at the boundary.

 Now we can compute the area derivative of the Dirichlet functional for an arbitrary parameterization of the boundary curve.
 
 The area derivative of this functional, as defined in \eqref{areaderDef}  would  yield exactly zero, as the second functional derivative is a symmetric tensor proportional to $\delta_{\mu\nu}$. The minimization over parametrizations would not change this fact, as we get zero for arbitrary parametrization. The area derivative of the minimal value of the Douglas-Dirichlet functional corresponds to the area derivative of this functional, taken at the optimal parametrization. This area derivative would, therefore, yield zero.

We conclude that the area derivative of the holomorphic functional yields zero, contradicting exact computations for the full area functional. 

So, the area derivative, for some reason, involves more than just the minimal value of the area considered as a functional of its boundary. This paradox is discussed in the next section.
\subsection{The hidden variables of the minimal surface}
\label{hidden}

The apparent paradox is resolved by examining the variational principle in the extended space. The functional $S_\chi[C]$ is defined as the minimum of the area functional for a surface $X^a_\mu \in \mathbb{R}^3 \otimes \mathbb{R}^4$. Crucially, the boundary of this extended surface is rigidly locked to the physical loop $C_\mu$ via the chiral 't Hooft symbols:
\begin{smalleq}
\begin{smalleq}\begin{align}
X^a_\mu\big|_{\partial} = \eta^{\chi, a}_{\mu\nu} C_\nu
\end{align}\end{smalleq}
\end{smalleq}
When we compute the area derivative, we vary the physical loop $C_\mu$ in standard 4D space. However, due to the boundary coupling, this 4D variation \textbf{induces} a specific variation of the surface in the 12-dimensional embedding space. We are probing the shape of this 12-dimensional surface in the vicinity of its 4-dimensional edge.

By the Envelope Theorem (or the principle of stationary action), the variation of the minimal value with respect to the boundary parameters is determined solely by the boundary terms, as the bulk Euler-Lagrange variations vanish. Thus:
\begin{smalleq}
\begin{smalleq}\begin{align}
\delta S_\chi = \oint d\theta \, \frac{\delta \mathcal{L}}{\delta (\partial_\theta X^a_\mu)} \delta X^a_\mu = \oint d\theta \, \Pi^a_\mu (\eta^{\chi, a}_{\mu\nu} \delta C_\nu)
\end{align}\end{smalleq}
\end{smalleq}
While the scalar minimal value of the action (the energy) is invariant under the isometric embedding $\eta$, the conjugate momentum $\Pi^a_\mu$ and the boundary coupling $\eta$ carry the chiral index. The first variation $\delta S$ vanishes at the minimal surface , but the area  derivative involves the second variation, related to variation of this conjugate momentum by $C'$. So to say, we are measuring the elasticity of the spring , which pulls the surface in 12D to its minimal shape. This minimal shape, does not depend on chirality but the elasticity tensor does.

The same arguments explain the difference between the area derivative of the ordinary Nambu-Goto area functional and the Dirichlet functional. The area derivatives uses the Hessian of the functional at the minimum, which is different for the full functional and the holomorphic projection, minimizing this full functional. This difference is responsible for the nonzero value of area derivative for the Nambu-Goto area functional

The analogy would be the computation of the hessian of the function of several variables at the extremal set, where the gradient of this function vanishes.
Neither the extremal value of the function, nor this vanishing gradient depend on some hidden parameters, but the hessian does. 
\begin{remark}[Toy example: same minimum value, different second variations]
Let
\[
F_a(x,y)=(y-a x)^2+x^2 \qquad (a\in\mathbb{R}).
\]
Then $y_\star(x)=ax$ minimizes in $y$ and the minimized value is
\[
g_a(x):=\min_y F_a(x,y)=F_a(x,ax)=x^2
\] independent of $a$.
Yet the $x$--Hessian of the \emph{full} functional on the minimizing locus depends on $a$:
\[
\left.\frac{\partial^2 F_a}{\partial x^2}\right|_{y=ax}=2(1+a^2)\neq g_a''(x)=2.
\]
\end{remark}

\subsection{Symmetrization of the area by parity}
The extremal value
\[
S_\chi[C]\;:=\;\min_{X}\,\mathcal{S}_\chi[X;C]
\]
is $\chi$-independent, but the loop-space area derivative---defined as the
discontinuity of the second variation and evaluated on the $\chi$-dependent
stationary surface in the extended space---retains $\chi$-dependence through
the boundary locking.
This chirality independence of the minimal value is manifest in the Holomorphic Ansatz: neither the Dirichlet functional \eqref{Dirichlet}, nor the Virasoro constraint \eqref{Virasoro}, nor the boundary condition \eqref{BCf} depend on $\chi$.  This common minimal value corresponds to a particular solution of the Plateau problem—the one which is additive over disconnected parts of the loop, including the parts arising at self-intersections. In the confining factor below, we preserve the parity of QCD by using a symmetric combination:\begin{smalleq}\begin{align}\label{areaSymm}S[C] = \frac{1}{2}( S_+[C] + S_-[C])\end{align}\end{smalleq}Its value equals the common value of these two functionals, but each term represents a different zero mode of the loop operator $\cL$ in \eqref{loopEq}, distinguished by Hodge duality $\chi = \pm 1$.The area derivative of the symmetrized area \eqref{areaSymm} is proportional to that of the Nambu-Goto string. By summing the self-dual and anti-self-dual projections:
\begin{smalleq}\begin{align*}
\frac{\delta S}{\delta \sigma} &= \frac{1}{2}\left(2 \frac{F + \ast F}{|F + \ast F|} + 2\frac{F - \ast F}{|F - \ast F|} \right) \\
&= \frac{2 F}{\sqrt{F^2 + (\ast F)^2}} = \frac{F \sqrt{2}}{|F|}
\end{align*}\end{smalleq}
Note that for the standard Nambu-Goto string, this factor would simply be $\frac{2F}{|F|}$. The difference is a factor of $\sqrt{2}$.

One may then ask: why do we need the Hodge-dual surface at all? If the area derivative of the symmetrized Hodge-dual surface is equivalent (up to normalization) to that of the Nambu-Goto string, why not claim the Nambu-Goto string itself solves the loop equation as a zero mode? Direct calculation in \cite{migdal2025SQYMflow} shows this is not possible. The area derivative of the conventional minimal area is well-defined, but it does not solve the loop equation: it produces in this equation a singular factor $\frac{\delta}{\delta^2 + \epsilon^2}$, depending on the ratio of the two infinitesimal cutoff parameters involved in the definition of the loop variations. Depending on this ratio, we obtain an arbitrary number between zero and infinity instead of the required zero in the loop equation.

The only way to satisfy the loop equation is to return to the symmetric sum of the \emph{full functionals} \eqref{areaSymm}, apply the area derivative to each term, and utilize the crucial identity that the sum of the derivatives relates to the difference of their duals:
\begin{smalleq}\begin{align*}
\frac{\delta S}{\delta \sigma} &= \ast \frac{\delta S_+}{\delta \sigma} - \ast \frac{\delta S_-}{\delta \sigma},\\
\frac{\delta S_\chi}{\delta \sigma} &= -2 (F + \chi \ast F)/\sqrt{ 2 F^2};\\
F_{\mu\nu} &= 2 \Im f_\mu \bar f_\nu
\end{align*}\end{smalleq}
We recall that each of these individual area derivatives $\frac{\delta S_\chi}{\delta \sigma}$ is the antisymmetric tensor part of the second derivative $\frac{\delta^2 S_{\pm}}{\delta C'_\mu(\theta-) \delta C'_\nu(\theta+)}$. We can then apply the Bianchi identity \eqref{Bianchi} to each term independently to satisfy the loop equation.
At the technical level, the dot derivative $\ff{C'_\mu}$ when applied to $(F_{\mu\nu} + \chi \ast F_{\mu\nu})/\sqrt{ F^2}$ yields this singular factor $\frac{\delta}{\delta^2 + \epsilon^2} \dot C_\nu$, but when the area derivative is represented exactly as the extremal value of $\chi \ast\frac{\delta^2 S_{\chi}}{\delta C'_\mu(\theta-) \delta C'_\nu(\theta+)}$, the same dot derivative $\ff{C'_\mu}$, taken before projecting on the holomorphic minimizer, yields zero by Bianchi identity.

In other words, the value of the area derivative of the symmetric combination $S_+ + S_-$ is not sufficient to compute the third derivative $\frac{\delta}{\delta C'}$ involved in the loop equation: one must retain the full functional definitions of the components $S_\pm$. The Bianchi identity allows us to bypass the laborious direct computation of this third functional derivative by exploiting the Hodge dualities of the two area derivatives. However, were we to perform this direct computation without invoking the Bianchi identity—retaining the full structure of $S_\pm$—we would find that the net contribution to the loop equation vanishes regardless. The singularities that plague the Nambu-Goto string, as well as the regular terms, must cancel exactly between the two chiral sectors to satisfy the Bianchi identity.

As an example of such a cancelation, one could consider a similar cancelation of Feynman diagrams in the expectations of $\hat P [D_\mu , \star F_{\mu\nu}]  \exp{\int A d x}$. With dimensional regularization, these diagrams cancel including regular and singular parts (poles in $\epsilon = 4-d$).
\subsection{Planar Loop Solution}

We consider the specific ansatz for a planar loop where the derivative of the coordinate vector, $f'_\mu(\xi)$, is defined by a single scalar holomorphic function $w'(\xi)$ multiplied by a constant isotropic vector:

\begin{smalleq}     \begin{align}
f'_\mu(\xi) = w'(\xi) \begin{pmatrix} \frac{1}{\sqrt{2}} \\ \frac{-\I}{\sqrt{2}} \\ 0 \\ 0 \end{pmatrix}_\mu
 \end{align} \end{smalleq}

\subsection{Verification of Null Condition}
The null condition $\left( f'(\xi) \right)^2 = 0$ is satisfied identically due to the vector structure:

\begin{smalleq}     \begin{align}
\sum_{\mu=1}^4 f'_\mu f'_\mu = (w'(\xi))^2 \frac{1^2 + (-\I)^2 + 0^2 + 0^2}{2}  = 0
 \end{align} \end{smalleq}

\subsection{Coordinate Integration}
The spacetime coordinates $X^a_\mu(\xi)$ are recovered by integrating the holomorphic derivatives. Assuming the physical coordinates are the real part of the holomorphic trajectory $X^a_\mu(\xi) =2\eta^{a,\chi}_{\mu\nu}\Re f_\nu(z)$:

\begin{smalleq}    \begin{align}
\Re f_1(z) &= \Re [w(z)] \\
\Re f_2(z)  &= \Re [-\I w(z)]  =\Im [w(z)] \\
\Re f_3(z) &= 0 \\
\Re f_4(z) &= 0
\end{align}\end{smalleq}

\subsection{Resulting Geometry}

This solution represents (up to rotation with $\eta^a$ matrices) a general conformal mapping from the parameter domain $\xi$ directly to the physical plane $(X_1, X_2)$, with trivial coordinates in the transverse dimensions $X_3, X_4$. The function $w(\xi)$ completely determines the shape of the planar loop.
This planar solution matches the one found in \cite{migdal2025SQYMflow} for planar loop.

\paragraph{Verification of the Loop Equation.}
For this planar solution, we can explicitly verify the zero-mode condition. The area derivative for the planar ansatz reduces to the simple algebraic form
\begin{smalleq}     \begin{align}
\label{circelArder}
\delta S_\chi/\delta \sigma_{\mu\nu} = -2 (F_{\mu\nu} + \chi \ast F_{\mu\nu})/ \sqrt{2F^2}
\end{align} \end{smalleq}
which is  manifestly self-dual. Therefore, the condition $*\frac{\delta S_\chi}{\delta \sigma} = \chi\frac{\delta S_\chi}{\delta \sigma}$ holds pointwise. By the Bianchi identity (2.5), this implies $\mathcal{L}S[C] \equiv 0$ without reliance on the general loop calculus of Ref [4]. This confirms that the explicit planar solution is indeed an exact zero-mode of the loop operator. This is, however, a rather trivial zero mode, as the normalized tensor $ (F_{\mu\nu} + \chi \ast F_{\mu\nu})/ \sqrt{2F^2}= \eta^{3,\chi}_{\mu\nu}$ is a universal constant tensor, with vanishing functional derivative.

\section{Topological Stability and General Variational Principle}

\subsection{Stationarity and Boundary Variation}
Let $C$ be a self-intersecting loop. The Hodge surface minimization problem admits an 
\textbf{additive solution} $\Sigma_{add}$ which is a stationary point 
(local minimum) of the geometric Area functional $S[X]$ defined on the 
extended space of maps $X: \D \to \mathbb{R}^3 \otimes \mathbb{R}^4$.

Because $\Sigma_{add}$ is a stationary point, the coordinate functions 
$X_\mu^a(u,v)$ satisfy the general Euler-Lagrange equations for the 
area density $\mathcal{L} = \sqrt{\Sigma^2/2}$.
This condition ensures that for \textit{any} variation $\delta X_\mu^a$ 
in the wider functional space (not limited to holomorphic functions), 
the bulk contribution to the variation of the action vanishes:
\begin{smalleq}\begin{align}
    \delta_{bulk} S = -\int_D d^2 \xi \, \delta X_\mu^a 
    \left[ \text{E-L Equations} \right]_\mu^a .
\end{align}\end{smalleq}
Consequently, the total variation of the area functional reduces 
exactly to the boundary flux term derived in \eqref{deltaSBC}.
This validates the computation of the area derivative and  subsequent proof of its Hodge-duality for arbitrary loop as long as the surface $X^a_\mu \in \R^3 \otimes \R^4$  minimizes the area functional \eqref{areaDef}.

The existence of the additive minimal solution is guaranteed by the 
results of  \cite{Douglas1931} and \cite{Rado1930}, who established that the 
disconnected solution minimizes the functional for separated or 
touching boundaries. Since the surface is an extremum in the full 
functional space, the bulk variation vanishes identically by the 
principle of stationary action.

\begin{remark}[Stability of the Goldschmidt Solution]
While a connected minimal surface (global minimum) may exist, the QCD 
Loop Equation requires the additive Goldschmidt solution to satisfy 
the factorization condition $S[C_1 \! \! \circ C_2] = S[C_1] + S[C_2]$. The topological barrier proven by
\cite{Gulliver1973} ensures that this additive solution is a local minimum, which is all we need for the proof of Hodge-duality. 
\end{remark}
\begin{remark}
The continuity of the transition between a single loop with an infinitesimal neck and two touching loops (the self-intersecting limit) is mathematically grounded in the \textit{Bridge Principle} for stable minimal surfaces \cite{White1996}. This principle establishes that given two stable minimal surfaces bounded by curves $C_1$ and $C_2$, one can construct a connected minimal surface bounded by a single loop formed by joining $C_1$ and $C_2$ with a thin ``bridge'' of width $\epsilon$. As $\epsilon \to 0$, this connected surface converges smoothly (in the varifold sense) to the union of the original surfaces. Consequently, the \textit{Goldschmidt solution} we utilize---which effectively treats the area as additive, $\mathcal{A}(C_1 \cup C_2) = \mathcal{A}(C_1) + \mathcal{A}(C_2)$---represents the correct continuous limit of the loop equation dynamics, avoiding the discontinuous jumps associated with the global minimization of the Plateau problem (such as the catenoid collapse). This ensures that the loop derivative remains well-defined across topology changes. Specifically, because the convergence is smooth in the varifold sense, the functional derivatives required for the loop equation (which probe the surface near the loop $C$) are continuous limits of the derivatives of the connected surface, justifying the application of the differential loop calculus to the additive Goldschmidt functional.
\end{remark}
We assume (or conjecture, supported by \cite{Douglas1931,Rado1930,White1996}) that the Goldschmidt/local-minimum branch defines a continuous additive functional on immersed loops with transverse self-intersections, and that the loop-space area derivative exists almost everywhere along each branch. Under this assumption, Theorem \ref{lem:finiteNc-dressing} applies.

\section{Confinement theorem}\label{sec:confinement}
Now, we can summarize the arguments of the previous section in a theorem.
\begin{theorem}[Confinement area law]\label{thm:Confinement}
The common minimal value of self-dual (and anti-dual) area $S_\chi[C]$for any loop without self-intersections is equal to the ordinary euclidean minimal area bounded by the same loop 
\begin{smalleq}     \begin{align}
S_{\chi}[C] = 2\sqrt{2} \, \mathcal{A}_{\text{Plateau}}[C]
\label{areaLaw}
 \end{align} \end{smalleq}
At self-intersections or for several disconnected closed components of the loop the area is additive over closed parts, unlike the euclidean minimal area which may change topology for a smaller area.
\begin{proof}
    The minimal value of Self-Dual Area $S_{\chi}$ is given by the Dirichlet integral \eqref{Dirichlet}  with boundary conditions $$ 2 \Re f_\mu(e^{\I \theta}) = C_\mu(\theta)$$
    under the Virasoro constraint \eqref{Virasoro}. These are exactly the  conditions of the classical Plateau problem. Therefore, every local minimum of our problem is also the local minimum of the Plateau problem.   The difference between solutions could arise in case of self-intersecting loop when there could be several local minima. Our solution selects the one that is additive over the closed parts of the loop (the Goldschmidt solution), whereas the Plateau problem selects the global minimum, which could be topologically different (say, a cylindrical surface bounded by both parts of the loop). According to the theorems \cite{Rado1930, Douglas1931, Gulliver1973}, such a disconnected, locally minimal surface always exist.
\end{proof}
\begin{remark}
    The equality of the areas $S_\pm[C]$ only holds for the minimal values. These areas are different as functionals, which leads to a difference in the second order functional derivatives at the minimum.
    As a result, the area functional $S_{\chi}[C]$ has the area derivatives with Hodge duality $\chi = \pm 1$ despite having common minimal value. 
    This phenomenon was discussed in section \ref{hidden} above.
\end{remark}
\end{theorem}
\begin{remark}
\textbf{(Distinction from String Theory and Lattice Strong Coupling).}
Unlike the lattice QCD strong-coupling expansion, there are no intrinsic ``worldsheet'' modes in our rigid minimal surface to generate daughter trajectories. While strong-coupling expansions in lattice gauge theory do involve sums over tessellated surfaces, these are essentially lattice artifacts that do not survive the continuum limit. It is well known that a sum over random surfaces does not exist in the 4D continuum theory due to conformal anomalies (which lead to the branched polymer instability); these anomalies are often undetectable in coarse lattice simulations. Consequently, using the lattice strong-coupling phase as a guide for the continuum theory is misleading, as it fails to respect fundamental continuum symmetries---specifically the conformal anomalies and the Hodge duality required for the mechanism described here. Our solution operates in the continuum limit where these constraints enforce a unique, rigid minimal surface rather than a sum over fluctuating surfaces.
\end{remark}
\begin{remark}
    We must also specify the definition of the fluctuation factor $W_{fluc}[C]$ in the Wilson loop. It is uniquely defined as an asymptotic expansion generated by iterating the MM equations [3, 4], starting from $W=1$. The physical hypothesis of this paper is that the actual gluonic vacuum of QCD is described by this dressed solution, choosing the fluctuation factor $W_{fluc}$ as the seed. Crucially, while the rigid minimal surface $S[C]$ provides the exact ground-state area law (confinement), $W_{fluc}[C]$ encapsulates all dynamic transverse gluon-exchange fluctuations. Therefore, physical phenomena such as the L\"uscher term ($-\pi/24R$) and the spectrum of excited flux tubes (hybrid mesons) observed in Lattice QCD are entirely generated by $W_{fluc}[C]$, perfectly complementing the geometric rigidity of the background vacuum. The ambiguity associated with adding the dressing factor with an arbitrary parameter $\kappa$ corresponds to the inherent non-perturbative ambiguity (Borel renormalons) in perturbative QCD.
\end{remark}
\subsection{The Zero Mode as the Resummed Multi-Instanton Vacuum}

It is crucial to emphasize the physical meaning of the confining factor
$\exp{-\kappa S[C]}$ as an exact zero mode of the loop diffusion operator.
The Makeenko-Migdal loop equations represent the exact quantum equations of
motion for Planar QCD. However, like any equations of motion, they do not
uniquely dictate the physical vacuum state of the theory. The properties of
the physical vacuum must be supplied independently, appearing mathematically
in the loop calculus as the choice of a specific zero-mode background.

In standard perturbation theory, one implicitly chooses the trivial
perturbative vacuum by iterating the loop hierarchy starting from the seed
$W_{fluc}[C] = 1$. While this successfully generates massless gluon exchange
diagrams, it fails to capture confinement. To describe the true physical
vacuum of QCD, the equations of motion must be solved in the presence of a
nonperturbative background.

In the traditional gauge-field representation, this nonperturbative vacuum is
envisioned as a complex mixed state---a dense ensemble or ``liquid'' of
multi-instantons and anti-instantons. A fundamental difficulty in this
coordinate-space representation is that any specific multi-instanton
configuration $A_{\mu}^{cl}(x)$ explicitly breaks macroscopic translation
invariance by localizing energy at the instanton centers. To recover the
true, translation-invariant physical QCD vacuum, one must perform an
exceedingly complicated superposition (integration) over the
infinite-dimensional moduli space of all multi-instanton sizes, positions,
and color orientations. In coordinate space, this summation is analytically
intractable.

This is where the profound power of the loop space representation reveals
itself. The mathematical condition for our loop-space zero mode is the
self-duality (or anti-self-duality) of the area derivative, which is the
exact geometric analogue of the $F_{\mu\nu} = \pm *F_{\mu\nu}$ condition for
instantons. However, because our minimal area functional $S_{\chi}[C]$
depends only on the relative geometry of the boundary loop $C$ and is
invariant under rigid translations, the exact translational invariance of the
vacuum is preserved from the outset.

Consequently, the horrifyingly complex topological ensemble of
multi-instantons in the gauge-field representation geometrically ``resums''
and simplifies in loop space into a single, exact, translation-invariant
object: the Hodge-dual minimal surface. The Area Law is therefore not an ad
hoc deformation, but rather the exact mathematical zero mode corresponding
to this physical mixed-state vacuum. The string tension parameter $\kappa$,
which factors out of the equations of motion entirely, acts as the
macroscopic order parameter (the gluon condensate) emerging from this
microscopic ensemble. In the next Section we compare the small-loop behavior
of this dressed solution with the Operator Product Expansion to explicitly
verify this relationship.

Subject to these caveats, we have established a precise correspondence between QCD and the area law. As demonstrated in the next paper \cite{Migdal2026GeometricQCDII}, this framework yields the standard rising Regge trajectories while avoiding the typical inconsistencies associated with string theory in four dimensions. Our approach relies solely on minimal surfaces—a classical and well-defined subject in modern geometry. 
While the mathematical theory of minimal surfaces as solutions to the Plateau problem is extensive, it generally does not provide explicit analytic solutions for an arbitrary loop $C$.
Fortunately, the main application of our minimal surface -- the rising meson spectrum in QCD-- does not require such a general solution. The relevant loops in this case, as we shall see in the \cite{Migdal2026GeometricQCDII} of this series, simplify, allowing for the analytic solution for the minimal surface and directly leading to linear Regge trajectories.

\section{String Scale $\kappa$ and Gluon Condensate}

The string mass scale $\kappa$ is not fixed by the MM equation itself. As we 
have seen, the confining factor $\exp{-\kappa S[C]}$ passes through the equation 
without correlation with the perturbative solution $W_{pert}[C]$. This factor 
reflects the non-perturbative properties of the QCD vacuum rather than the 
perturbative fluctuations of the gluon field.

The fundamental non-perturbative property of the QCD vacuum is the gluon 
condensate~\cite{Shifman:1978bx}. We relate the string tension $\kappa$ to 
this condensate by considering a circular loop $C$ of small radius $r$. 

According to our solution, for a circular loop of radius $r$ one has
$S[C]=2\sqrt{2}\,\pi r^{2}$, and we evaluate the second area derivative at the
two opposite points $x_{1}=C(0)$ and $x_{2}=C(\pi)$, with $|x_{12}|=2r$.
In asymptotically free QCD the OPE gives, at short distances,
\begin{smalleq}
\begin{align}
&\frac{\delta^{2} W_{\mathrm{pert}}[C]}{\delta\sigma_{\mu\nu}(1)\,\delta\sigma_{\lambda\rho}(2)}
\propto
\big\langle \mathrm{Tr}\,F_{\mu\nu}(x_{1})\,U_{0,\pi}\,F_{\lambda\rho}(x_{2})\,U_{\pi,2\pi}\big\rangle \nonumber\\
&\to \frac{\alpha_{\mathrm{eff}}}{4\pi}\,
\frac{ I_{\mu\lambda}I_{\nu\rho}-I_{\nu\lambda}I_{\mu\rho} }{|x_{12}|^{4}}\,
W_{\mathrm{pert}}[C] \;+\;\cdots ,
\label{OPE}
\end{align}
\end{smalleq}
where $I_{\mu\lambda}=\delta_{\mu\lambda}-2\hat x_{\mu}\hat x_{\lambda}$ and the ellipsis
denotes terms less singular as $|x_{12}|\to 0$ (including local condensates).
We only track the tensor structure and the first nonvanishing local operator after projection; overall Wilson-coefficient conventions match SVZ \cite{Shifman:1978bx}.
To isolate the nonperturbative dimension--4 contribution, we contract the
antisymmetric index pairs with the projector
$\Pi_{\mu\nu\lambda\rho}=\frac12(\delta_{\mu\lambda}\delta_{\nu\rho}-\delta_{\nu\lambda}\delta_{\mu\rho})$.
A direct contraction shows that the leading conformal structure is annihilated:
\[
\Pi_{\mu\nu\lambda\rho}\big(I_{\mu\lambda}I_{\nu\rho}-I_{\nu\lambda}I_{\mu\rho}\big)=0 .
\]
Therefore the projected OPE starts at dimension~4, and the leading surviving term is
the gluon condensate,
\begin{smalleq}
    \begin{align}
&\Pi_{\mu\nu\lambda\rho}\,\mathrm{Tr}\,F_{\mu\nu}(1)\,U_{1,2}\,F_{\lambda\rho}(2)\,U_{2,1}\to\br
&-\frac{g^{2}}{2N_{c}}\langle (G^{a}_{\mu\nu})^{2}\rangle\,W_{\mathrm{pert}}[0] .
\label{PIII}
 \end{align}
\end{smalleq}

To verify this projection, define the scalar contraction
\begin{smalleq}
    \begin{align}
    \label{Seq}
&S(d)\;:=\;\Pi_{\mu\nu\lambda\rho}\Big(I_{\mu\lambda}I_{\nu\rho}-I_{\nu\lambda}I_{\mu\rho}\Big),
\br
&\Pi_{\mu\nu\lambda\rho}=\tfrac12(\delta_{\mu\lambda}\delta_{\nu\rho}-\delta_{\nu\lambda}\delta_{\mu\rho}),
\br
&I_{\mu\lambda}=\delta_{\mu\lambda}-2\hat x_\mu\hat x_\lambda ,
 \end{align}
\end{smalleq}
in Euclidean dimension $d$ (so $\delta_{\mu\mu}=d$ and $\hat x^2=1$). Expanding the contraction gives
\begin{smalleq}
    \begin{align}
S(d)
&=\oh(\delta_{\mu\lambda}\delta_{\nu\rho}-\delta_{\nu\lambda}\delta_{\mu\rho})
\big(I_{\mu\lambda}I_{\nu\rho}-I_{\nu\lambda}I_{\mu\rho}\big)\\
&=
\oh\lrb{\underbrace{\delta_{\mu\lambda}\delta_{\nu\rho}I_{\mu\lambda}I_{\nu\rho}}_{=(\mathrm{tr}\,I)^2}
-\underbrace{\delta_{\mu\lambda}\delta_{\nu\rho}I_{\nu\lambda}I_{\mu\rho}}_{=\mathrm{tr}(I^2)}}\br
&+ \oh\lrb{-\underbrace{\delta_{\nu\lambda}\delta_{\mu\rho}I_{\mu\lambda}I_{\nu\rho}}_{=\mathrm{tr}(I^2)}
+\underbrace{\delta_{\nu\lambda}\delta_{\mu\rho}I_{\nu\lambda}I_{\mu\rho}}_{=(\mathrm{tr}\,I)^2}}
\\
&=\mathrm{tr}\,I)^2-\mathrm{tr}(I^2).
\end{align}
\end{smalleq}
Now $\mathrm{tr}\,I=I_{\mu\mu}=d-2\hat x_\mu\hat x_\mu=d-2$, and, writing $I=\mathbf{1}-2\hat x\hat x^{T}$,
\begin{smalleq}
    \begin{align}
&I^2=(\mathbf{1}-2\hat x\hat x^{T})^2=\mathbf{1}-4\hat x\hat x^{T}+4\hat x(\hat x^{T}\hat x)\hat x^{T}=\mathbf{1}
\br
&\Rightarrow\quad \mathrm{tr}(I^2)=d.
\end{align}
\end{smalleq}
Therefore
\[
S(d)=(d-2)^2-d=(d-1)(d-4),
\]
and in particular $S(4)=0$, i.e. the leading conformal tensor structure in \eqref{Seq} is annihilated by the
$\Pi_{\mu\nu\lambda\rho}$ projection specifically in four dimensions.

The remaining term in OPE is the gluon condensate contribution. 
\begin{remark}
This projection removes the \emph{universal leading} short-distance conformal tensor structure
$\propto |x_{12}|^{-4}$ in~\eqref{Seq}. Subleading perturbative contributions and higher-dimensional
operators are not claimed to vanish; they are less singular as $|x_{12}|\to 0$ and are suppressed
by additional powers of the loop size in the small-loop expansion.
\end{remark}
\begin{smalleq}
    \begin{align}
   & \Pi_{\mu\nu\lambda\rho} \mathrm{Tr} F_{\mu\nu}(1) U_{1,2} F_{\lambda\rho}(2) U_{2,1} 
    \to\nonumber\\
    &-\frac{g^2}{2 N_c}\langle (G^a_{\mu\nu})^2 \rangle W_{pert}[0] .
    \end{align}
\end{smalleq}
This constant contribution from the gluon condensate must be balanced by the second area derivative of 
our Hodge-dual factor in the ansatz $W = W_{pert} e^{-\kappa S}$. Using the 
previously found area derivative \eqref{circelArder} for a circle , we find:
\begin{smalleq}
    \begin{align}
       & W_{pert}[C]\frac{\delta^2 \exp{-\kappa S[C]}}{\delta \sigma_{\mu\nu}(0)\delta \sigma_{\lambda\rho}(\pi)} 
        \to \nonumber\\
        &4\kappa^2 W_{pert}[0] \Pi_{\mu\nu\lambda\rho}\eta^{3,\chi}_{\mu\nu}\eta^{3,\chi}_{\lambda\rho} = 16 \kappa^2W_{pert}[0]
    \end{align}
\end{smalleq}
Comparing these terms, we find (in conventional normalizarion)
\begin{smalleq}
    \begin{align}
 \kappa^2 = \frac{\pi^2}{8 N_c}\VEV{\frac{\alpha_s}{\pi}(G^a_{\mu\nu})^2}
 \end{align}
\end{smalleq}
\begin{remark}
    All quantities in the OPE are understood at a renormalization scale $\mu\sim 1/r$.
Accordingly, the Wilson coefficients, $\alpha_s(\mu)$, the condensate
$\langle (G_{\mu\nu}^a)^2\rangle_\mu$, and the inferred parameter $\kappa(\mu)$ are
scheme- and scale-dependent; this is consistent with the intrinsic nonperturbative ambiguity
in the definition of $W_{\mathrm{pert}}[C]$ discussed in previous remarks.
\end{remark}
\begin{remark}
\noindent
In our conventions the area-law exponent is $\kappa S[C]$. Using~\eqref{areaLaw}, for a planar loop
with physical Plateau area $A_{\mathrm{Pl}}[C]$ one has $S[C]=2\sqrt{2}\,A_{\mathrm{Pl}}[C]$,
so the effective string tension multiplying $A_{\mathrm{Pl}}$ is
\[
\sigma \;=\; 2\sqrt{2}\,\kappa\,.
\]
\end{remark}

Our association of $\kappa$ with the gluon condensate supports the physical 
picture of condensed electric flux originally conjectured by Nambu~\cite{Nambu:1974zg}. 
One may view the term $\kappa S[C]$ as the action cost of a field configuration 
with energy density $\langle G^2 \rangle \sim \Lambda_{QCD}^{4}$ spread over 
a ``pancake'' surface of thickness $\Lambda_{QCD}^{-2}$. This estimate interprets 
the dressing factor $\exp{-\kappa S[C]}$ as a statistical probability 
$\exp{-\beta \delta E}$ for a large quark loop propagating through the QCD 
vacuum (the dual Meissner effect).

This computation, combined with asymptotic freedom, established the correct scaling of the string tension $\kappa \propto \Lambda_{QCD}^2$.
\section{Discussion: Comparison with other models of Confinement}

A fundamental distinction must be drawn between the geometric solution presented here and the traditional strong coupling expansion in Lattice Gauge Theory (LGT).

In the standard LGT framework, the area law arises in the strong coupling limit ($\beta \to 0$) through a sum over tessellated surfaces (plaquettes) bounded by the Wilson loop $C$. While this yields confinement at strong coupling, the analytic continuation of this expansion to the continuum limit ($\beta \to \infty$) is obstructed by two well-known pathologies: the \emph{Roughening Transition} and the \emph{Branched Polymer instability}.

\subsection{The Failure of Surface Sums in the Continuum}
As the coupling decreases towards the continuum limit, the surface tension decreases, and the ``tessellated surface'' undergoes a Roughening Transition \cite{Hasenfratz1981, Luscher1981}. Beyond this point, the transverse fluctuations of the surface diverge, and the simple strong coupling expansion breaks down.

More critically, attempts to formulate the Wilson loop as a sum over random surfaces in continuous 4D Euclidean space encounter a fatal geometric instability. The entropy of ``crumpled'' surface configurations overwhelms the Boltzmann weight of the area action. Consequently, the path integral over surfaces is dominated not by smooth 2D sheets, but by \textbf{Branched Polymers}---tree-like structures with Hausdorff dimension $d_H = 4$ \cite{Frohlich1985, Durhuus1984}. This ``crumpled phase'' is unphysical for a string theory of hadrons.

Furthermore, if one attempts to treat these surfaces as fundamental dynamical strings (e.g., the Polyakov string), the theory suffers from the \textbf{Conformal (Liouville) Anomaly}. In $D=4$, the trace anomaly of the stress tensor leads to a strongly coupled Liouville mode, preventing a consistent quantization of the string without introducing unphysical critical dimensions (such as $D=26$).

\subsection{Stability of the Hodge-Dual Solution}
The solution derived in this paper avoids these instabilities entirely because \textbf{we do not perform a sum over fluctuating metrics.}

Instead of integrating over a path space of random surfaces, our approach identifies a unique, stable geometric object: the \textbf{Hodge-dual minimal surface}. This surface is a classical saddle point of the effective action in loop space, determined rigidly by the boundary loop $C$ and the specific geometry of the Douglas-Gram matrix.
\begin{itemize}
    \item \textbf{No Roughening:} Because the surface is a fixed minimal solution (a ``hologram'' of the loop) rather than a fluctuating statistical ensemble, it does not undergo a roughening transition.
    \item \textbf{No Branched Polymers:} We do not sum over ``crumpled'' configurations; we select the analytic minimum.
    \item \textbf{No Liouville Anomaly:} There are no intrinsic string degrees of freedom (2D quantum gravity) to quantize. The surface is an auxiliary geometric construct derived from the gauge field dynamics, not a fundamental quantum string.
\end{itemize}

The construction requires the summation of the areas of the two Hodge-dual surfaces, $S$ and $\tilde{S}$. This is dictated by \textbf{Parity Conservation}: since the Hodge star operation involves the epsilon tensor, it transforms as a pseudoscalar. Including both surfaces symmetrically ensures the resulting Wilson loop functional is a true scalar, preserving the parity invariance of the strong interaction.

\subsection{Distinct from Nambu-Goto String Models}
It is crucial to emphasize that this solution does not merely reproduce the standard string model of confinement. The classic Nambu-Goto string spectrum predicts a tower of ``daughter trajectories'' arising from the vibrational modes of the flux tube.

In contrast, our solution yields a single linear Regge trajectory. The confining factor $\exp{-\kappa S_{min}}$ corresponds to the ground-state static potential $V_0(R)$ between quarks.

We do not deny the existence of excited gluonic levels (often termed ``hybrid potentials'' or ``excited flux tubes'' in lattice QCD \cite{Juge2003}). However, within our framework, these states are not vibrations of the confining minimal surface itself. The minimal surface is a rigid geometric object fixed by the loop $C$. Excited states correspond to specific operator insertions or gluonic excitations \emph{on top} of this vacuum background, which would be described by higher-order correlation functions rather than the vacuum expectation value of a single Wilson loop.

Thus, the ``missing'' daughter trajectories in our result are a feature, not a bug: we separate the universal confining force (the minimal surface) from the specific gluonic excitations, avoiding the unphysical tachyon/Liouville modes inherent to fundamental string theories in $D=4$.

\subsection{Decoupling of Confinement and Correlations}
This factorization property provides a crucial insight into the physical nature of the mass gap. Consider the connected correlator for two loops $C_1$ and $C_2$ separated by a large distance $R$:
\begin{smalleq}
\begin{align}
    \mathcal{W}_{conn}(C_1, C_2) = \mathcal{W}(C_1 \cup C_2) - \mathcal{W}(C_1)\mathcal{W}(C_2)
\end{align}
\end{smalleq}
In our framework, the minimal area functional is additive for widely separated loops, $S_{min}[C_1 \cup C_2] = S_{min}[C_1] + S_{min}[C_2]$. Consequently, the non-perturbative confining factor factors out of the correlation completely:
\begin{smalleq}
\begin{align}
    &\mathcal{W}_{conn}(C_1, C_2) = \exp{-\kappa (S[C_1] + S[C_2])} \br
    &\left[ W_{pert}(C_1, C_2) - W_{pert}(C_1)W_{pert}(C_2) \right]
\end{align}
\end{smalleq}
This result implies that the exponential decay of correlations at large distance (the glueball mass gap) is \textbf{not} driven by the confining minimal surface. The surface simply provides the background normalization for the quark sources (the confinement mechanism).

Instead, the inter-loop correlation $\exp{-M_{gl}R}$ arises entirely from the term in brackets---the connected part of the perturbative gluon diagrams (``fluctuation factors''). Thus, our solution establishes a clean separation of mechanisms:
\begin{enumerate}
    \item \textbf{Confinement} is a geometric phenomenon governed by the Hodge-dual minimal surface (the area law).
    \item \textbf{The Mass Gap} (glueballs) and long-range correlations are dynamic phenomena governed by gluon exchange fluctuations (the perturbative pre-factors).
\end{enumerate}
This distinction further explains why our minimal surface remains stable and static: it does not need to vibrate to transmit forces between singlets; those forces are carried by the gluon field itself. The problem of the mass gap in this correlation function is not solved by our confining factor: it requires further nonperturbative insights.

\subsection{Our 12D Surface vs. AdS/CFT and Holographic QCD}

Our solution shares the fundamental philosophy of the Holographic Principle: the physical observables (Wilson loops) in 4D spacetime are determined by the geometry of a minimal surface in a higher-dimensional "bulk" space. However, the nature of this target space and the mechanism of holography differ fundamentally from the standard AdS/CFT correspondence \cite{Maldacena1998} or Holographic QCD models \cite{Witten:1998zw, Sakai:2004cn}.

\begin{enumerate}
    \item \textbf{The Target Space:} In AdS/CFT and models like Sakai-Sugimoto \cite{Sakai:2004cn}, the
    holographic dual is a curved 5-dimensional anti-de Sitter space (plus a
    compact manifold), representing a dynamical gravitational theory. In our
    framework, the "bulk" space is a flat 12-dimensional space, $\mathbb{R}^{4}
    \otimes\mathfrak{su}(2)$, formed by the tensor product of Euclidean spacetime
    and the canonical $\mathfrak{su}(2)$ subalgebra associated with the
    self-dual sector. Because these extra dimensions represent an 'unfolded'
    internal color space rather than physical spacetime, there are no propagating
    bulk gravitons. Consequently, pure Yang-Mills strictly avoids backreaction
    on this flat target space.
    \item \textbf{Origin of the Geometry:} We do not postulate a dual gravitational theory. The 12-dimensional geometry arises directly from the non-Abelian structure of the gauge field itself. The minimal surface $X_\mu^a(\xi)$ is not a string moving in a gravitational background, but a geometric representation of the gauge field components "unfolded" into their internal color dimensions.
    \item \textbf{Conformal Invariance vs. Confinement:} Standard AdS/CFT begins
    with a conformal field theory ($\mathcal{N}=4$ SYM). To describe QCD, one must
    break this symmetry "by hand" via deformations of the metric, such as
    introducing a hard wall, a soft wall, or compactifying D4-branes on a thermal
    circle \cite{Witten:1998zw}. Our approach is the inverse: we begin with the
    QCD loop equation and observe a confining factor which is left undefined by
    this loop equation, but physically corresponds to the unknown QCD vacuum
    (with some gluon condensate in it). The Area Law then follows as a
    mathematical zero mode of the loop diffusion operator, corresponding to the
    physical QCD vacuum. The theory is inherently confining and non-conformal
    from the outset, with the mass scale set naturally by the string tension
    $\sigma$ (derived from the gluon condensate in Section 6), without ad hoc
    geometric deformations.
\end{enumerate}

Thus, our solution can be viewed as "Gauge Holography" rather than "Gravity Holography." The minimal surface is a hologram of the loop, projected not into a fictitious gravitational dimension, but into the internal color space of the QCD vacuum.

\subsection{Relation to Scattering Amplitude Duality}
Finally, we must address the well-known duality between Wilson loops and scattering amplitudes observed in $\mathcal{N}=4$ Super Yang-Mills theory \cite{Alday:2007hr, Drummond:2007aua}. In that context, the vacuum expectation value of a polygonal Wilson loop equals the MHV scattering amplitude (squared).

It is crucial to recognize that this duality relies entirely on the \textbf{dual conformal symmetry} of the planar $\mathcal{N}=4$ theory. In a conformal theory, the area law is strictly absent (or divergent), and the Wilson loop is dominated by cusp anomalies and perimeter terms.

In contrast, pure Planar QCD is \textbf{never conformal}. The running coupling constant violates conformal symmetry, leading to dimensional transmutation and the generation of a mass gap (as derived in Section 6). Consequently, the Wilson loop in pure QCD is dominated by the Area Law $\exp{-\sigma \text{Area}}$, which has no counterpart in the conformal scattering amplitudes of $\mathcal{N}=4$ SYM.

Furthermore, the existing literature on the Wilson loop/amplitude duality typically relies on the AdS/CFT conjecture without solving the underlying equations of motion for the gauge field. In this work, we have solved the \textbf{MM loop equations} \cite{MM1981NPB, Mig83} directly. Since these are the defining non-perturbative equations of the theory, our solution captures the specific dynamics of asymptotically free QCD, which are physically distinct from the kinematics of conformal Super Yang-Mills.

\section{Conclusions}

We have established the geometric foundation for the string dual of Planar QCD. The confinement mechanism relies fundamentally on the self-duality of the area derivative---a geometric property that exists exclusively in four dimensions. This stands in sharp contrast to \emph{fundamental} bosonic string models, which require unphysical critical dimensions (e.g., $D=26$) to maintain quantum consistency. While effective string theories can describe long flux tubes in $D=4$, our result shows that the fundamental confining mechanism is a stable, non-vibrating minimal surface, exempt from the quantization anomalies of the fundamental Nambu-Goto string.

Our results demonstrate that the confining surface is stable specifically in $D=4$. This dimensionality constraint explains the absence of daughter trajectories in the resulting spectrum: the confining object is a minimal surface fixed by 4D geometry, not a fundamental string vibrating in higher embedding dimensions. The ``missing'' string modes are eliminated, leaving a clean linear Regge trajectory consistent with phenomenological observations.

In the companion paper \cite{Migdal2026GeometricQCDII}, we populate this rigid surface with internal Majorana fermions. We show that this "Fermi String" theory reproduces the planar graphs of QCD and leads to a finite effective action in Momentum Loop Space, allowing for the computation of the meson spectrum.

\section*{Acknowledgements}
The author is grateful to Nima Arkani-Hamed for the invitation to present this work at the IAS Particle Theory "Pizza" Seminar, which provided a stimulating environment for discussing these results. I also thank Edward Witten for valuable discussions regarding the distinction between this confining geometry and the conventional area law.

\bibliographystyle{elsarticle-num}
\bibliography{bibliography}

@article{MMEq79,
title = {Exact equation for the loop average in multicolor QCD},
journal = {Physics Letters B},
volume = {88},
number = {1},
pages = {135-137},
year = {1979},
issn = {0370-2693},
doi = {url{https://doi.org/10.1016/0370-2693(79)90131-X}},
url = {url{https://www.sciencedirect.com/science/\\
           article/pii/037026937990131X}},
author = {Yu.M. Makeenko and A.A. Migdal}
}

@article{MM1981NPB,
  author       = {Yuri M. Makeenko and Alexander A. Migdal},
  title        = {Quantum Chromodynamics as Dynamics of Loops},
  journal      = {Nuclear Physics B},
  volume       = {188},
  pages        = {269--316},
  year         = {1981},
  doi          = {10.1016/0550-3213(81)90105-2},
  publisher    = {Elsevier},
}

@article{migdal2025SQYMflow,
title = {Spontaneous quantization of the Yang–Mills gradient flow},
journal = {Nuclear Physics B},
volume = {1020},
pages = {117129},
year = {2025},
issn = {0550-3213},
doi = {https://doi.org/10.1016/j.nuclphysb.2025.117129},
url = {https://www.sciencedirect.com/science/article/pii/S0550321325003384},
author = {Alexander Migdal}
}

@article{Mig83,
  author  = {A. A. Migdal},
  title   = {Loop equations and $1/N$ expansion},
  journal = {Physics Reports},
  volume  = {102},
  number  = {4},
  pages   = {199--290},
  year    = {1983},
  doi     = {10.1016/0370-1573(83)90076-5}
}

@misc{kazakovZheng2024,
    title={Bootstrap for Finite N Lattice Yang-Mills Theory},
    author={Vladimir Kazakov and Zechuan Zheng},
    year={2024},
    eprint={2404.16925},
    archivePrefix={arXiv},
    primaryClass={hep-th}
}

@article{Maldacena1998,
  author       = {Juan M. Maldacena},
  title        = {The Large N Limit of Superconformal Field Theories and Supergravity},
  journal      = {Advances in Theoretical and Mathematical Physics},
  volume       = {2},
  pages        = {231--252},
  year         = {1998},
  archivePrefix = {arXiv},
  eprint       = {hep-th/9711200},
  url          = {https://arxiv.org/abs/hep-th/9711200}
}

@article{Douglas1931,
  author    = {Jesse Douglas},
  title     = {Solution of the Problem of Plateau},
  journal   = {Transactions of the American Mathematical Society},
  volume    = {33},
  number    = {1},
  pages     = {263--321},
  year      = {1931},
  publisher = {American Mathematical Society},
  doi       = {10.2307/1989631},
  url       = {https://doi.org/10.2307/1989631}
}

@article{Nambu:1974zg,
    author = "Nambu, Yoichiro",
    title = "{Strings, Monopoles, and Gauge Fields}",
    journal = "Phys. Rev. D",
    volume = "10",
    pages = "4262--4268",
    year = "1974",
    doi = "10.1103/PhysRevD.10.4262",
    reportNumber = "COO-1764-202"
}

@article{Shifman:1978bx,
    author = "Shifman, M. A. and Vainshtein, A. I. and Zakharov, V. I.",
    title = "{QCD and Resonance Physics. Theoretical Foundations}",
    journal = "Nucl. Phys. B",
    volume = "147",
    pages = "385--447",
    year = "1979",
    doi = "10.1016/0550-3213(79)90022-1"
}

@article{poincare1907uniformisation,
  title={Sur l'uniformisation des fonctions analytiques},
  author={Poincar{\'e}, Henri},
  journal={Acta Mathematica},
  volume={31},
  number={1},
  pages={1--63},
  year={1907},
  publisher={Springer},
  doi={10.1007/BF02415442}
}

@article{Rado1930,
  author = {Rad{\'o}, Tibor},
  title = {On {Plateau's} problem},
  journal = {Annals of Mathematics},
  volume = {31},
  number = {3},
  pages = {457--469},
  year = {1930},
  publisher = {Mathematics Department, Princeton University},
  doi = {10.2307/1968237}
}

@article{Gulliver1973,
  author = {Gulliver, Robert},
  title = {Regularity of minimizing surfaces of prescribed mean curvature},
  journal = {Annals of Mathematics},
  volume = {97},
  number = {2},
  pages = {275--305},
  year = {1973},
  publisher = {Mathematics Department, Princeton University},
  doi = {10.2307/1970845}
}

@article{White1996,
  author    = {White, Brian},
  title     = {The Bridge Principle for Stable Minimal Surfaces},
  journal   = {Calculus of Variations and Partial Differential Equations},
  volume    = {4},
  number    = {5},
  pages     = {411--425},
  year      = {1996},
  publisher = {Springer}
}

@article{Hasenfratz1981,
  title = {Roughening transition in lattice gauge theories in arbitrary dimension},
  author = {Hasenfratz, Anna and Hasenfratz, Etelka and Hasenfratz, Peter},
  journal = {Nuclear Physics B},
  volume = {180},
  number = {3},
  pages = {353--367},
  year = {1981},
  publisher = {Elsevier}
}

@article{Luscher1981,
  title = {How thick are chromoelectric flux tubes?},
  author = {L{\"u}scher, Martin and M{\"u}nster, Gernot and Weisz, Peter},
  journal = {Nuclear Physics B},
  volume = {180},
  number = {1},
  pages = {1--12},
  year = {1981},
  publisher = {Elsevier}
}

@incollection{Frohlich1985,
  title = {The Statistical Mechanics of Surfaces},
  author = {Fr{\"o}hlich, J{\"u}rg},
  booktitle = {Applications of Field Theory to Statistical Mechanics},
  editor = {Garrido, L.},
  series = {Lecture Notes in Physics},
  volume = {216},
  pages = {31--57},
  year = {1985},
  publisher = {Springer, Berlin, Heidelberg}
}

@article{Durhuus1984,
  title = {Critical properties of constrained random surfaces},
  author = {Durhuus, Bergfinnur and Fr{\"o}hlich, J{\"u}rg and Jonsson, Thordur},
  journal = {Nuclear Physics B},
  volume = {240},
  number = {4},
  pages = {453--480},
  year = {1984},
  publisher = {Elsevier}
}

@article{Juge2003,
  title = {Fine Structure of the QCD String Spectrum},
  author = {Juge, K. J. and Kuti, J. and Morningstar, C. J.},
  journal = {Phys. Rev. Lett.},
  volume = {90},
  pages = {161601},
  year = {2003}
}

@article{Witten:1998zw,
    author = "Witten, Edward",
    title = "{Anti-de Sitter space, thermal phase transition, and confinement in gauge theories}",
    journal = "Adv. Theor. Math. Phys.",
    volume = "2",
    pages = "505--532",
    year = "1998"
}

@article{Sakai:2004cn,
    author = "Sakai, Tadakatsu and Sugimoto, Shigeki",
    title = "{Low energy hadron physics in holographic QCD}",
    journal = "Prog. Theor. Phys.",
    volume = "113",
    pages = "843--882",
    year = "2005"
}

@article{Alday:2007hr,
    author = "Alday, Luis F. and Maldacena, Juan Martin",
    title = "{Gluon scattering amplitudes at strong coupling}",
    journal = "JHEP",
    volume = "06",
    pages = "064",
    year = "2007"
}

@article{Drummond:2007aua,
    author = "Drummond, J. M. and Korchemsky, G. P. and Sokatchev, E.",
    title = "{Conformal properties of four-gluon planar amplitudes and Wilson loops}",
    journal = "Nucl. Phys. B",
    volume = "795",
    pages = "385--408",
    year = "2008"
}

@article{Migdal2026GeometricQCDII,
  author = {Migdal, Alexander},
  title = {Geometric {QCD} {II}: The Confining Twistor String and Meson Spectrum},
  journal = {Nuclear Physics B},
  year = {2026},
  note = {To be submitted}
}
\end{document}